\documentclass[letterpaper,11pt]{article}
\usepackage[margin=1in]{geometry}  

\usepackage{listings}
\usepackage{amsmath}
\usepackage{amsthm}
\usepackage{tikz}
\usetikzlibrary{positioning,matrix,decorations.pathreplacing,calc,topaths,arrows,automata}
\usepackage{caption,subcaption}
\usepackage{array}
\usepackage{mdwmath}
\usepackage{multirow}
\usepackage{mdwtab}
\usepackage{eqparbox}
\usepackage{multicol}
\usepackage{amsfonts}
\usepackage{multirow,bigstrut,threeparttable}
\usepackage{amsthm}
\usepackage{array}
\usepackage{bbm}
\usepackage{epstopdf}
\usepackage{mdwmath}
\usepackage{mdwtab}
\usepackage{eqparbox}
\usepackage{latexsym}
\usepackage{cite}
\usepackage{amssymb}
\usepackage{bm}
\usepackage{amssymb}
\usepackage{graphicx}
\usepackage{mathrsfs}
\usepackage{epsfig}
\usepackage{psfrag}
\usepackage{setspace}
\usepackage[
            CJKbookmarks=true,
            bookmarksnumbered=true,
            bookmarksopen=true,
            colorlinks=true,
            citecolor=red,
            linkcolor=blue,
            anchorcolor=red,
            urlcolor=blue
            ]{hyperref}
\usepackage[ruled]{algorithm2e}
\usepackage{algpseudocode}
\usepackage{stfloats}

\input xy
\xyoption{all}

\theoremstyle{plain}

\newtheorem{assumption}{Assumption}

\newtheorem{proposition}{Proposition}

\theoremstyle{definition}

\def\EE{\mathbb{E}}

\def\PP{\mathbb{P}}

\def\calA{\mathcal{A}}

\def\calN{\mathcal{N}}

\def\1{\mathbbm{1}}
\def\var{\mathsf{Var}}
\def\cov{\mathsf{Cov}}

\newcommand\independent{\protect\mathpalette{\protect\independenT}{\perp}}
\def\independenT#1#2{\mathrel{\rlap{$#1#2$}\mkern2mu{#1#2}}}

\usepackage{booktabs}
\usepackage{adjustbox}
\usepackage{multirow}
\usepackage{listings}

\theoremstyle{plain}

\def \var {\mathsf{Var}}

\usepackage{xspace}

\def\independenT#1#2{\mathrel{\rlap{$#1#2$}\mkern2mu{#1#2}}}

\definecolor{myblue}{rgb}{.8, .8, 1}
\definecolor{mathblue}{rgb}{0.2472, 0.24, 0.6} 
\definecolor{mathred}{rgb}{0.6, 0.24, 0.442893}
\definecolor{mathyellow}{rgb}{0.6, 0.547014, 0.24}

\usepackage{tabularx}
\newcommand{\RNum}[1]{\uppercase\expandafter{\romannumeral #1\relax}}

\DeclareMathOperator*{\argmax}{argmax}
\DeclareMathOperator*{\argmin}{argmin}

\usepackage{amsmath, amssymb, amsthm, amsfonts, amsxtra, bbm, mathrsfs, stmaryrd}

\usepackage{float}
\usepackage{nameref}

\tikzset{fontscale/.style = {font=\relsize{#1}}}

\begin{document}

\title{Assessment of Heterogeneous Treatment Effect Estimation Accuracy via Matching}
\author{Zijun Gao\thanks{Department of Statistics, Stanford University.}, 
Trevor Hastie\thanks{Department of Statistics and Department of Biomedical Data Science, Stanford University.},  
Robert Tibshirani\footnotemark[2]
}
\maketitle

\begin{abstract}
We study the assessment of the accuracy of heterogeneous treatment effect (HTE) estimation, where the HTE is not directly observable so standard computation of prediction errors is not applicable.
To tackle the difficulty, we propose an assessment approach by constructing pseudo-observations of the HTE based on matching.
Our contributions are three-fold: 
first, we introduce a novel matching distance derived from proximity scores in random forests; 
second, we formulate the matching problem as an average minimum-cost flow problem and provide 
an efficient algorithm;
third, we propose a match-then-split principle for the assessment with cross-validation.
We demonstrate the efficacy of the assessment approach on synthetic data and data generated from a real dataset.
\end{abstract}

\tableofcontents

\section{Introduction}\label{sec:introduction}
Nowadays the heterogeneous treatment effect (HTE) estimation under the Neyman-Rubin potential outcome model \cite{rubin1974estimating,splawa1990application} is gaining increasing popularity due to various practical demands, such as personalized medicine \cite{low2016comparing, lesko2007personalized}, personalized education \cite{murphy2016handbook}, and personalized advertisements \cite{bennett2007netflix}.
There are a number of works focusing on estimating the HTE using various machine learning tools: LASSO \cite{imai2013estimating}, random forests \cite{wager2018estimation}, boosting \cite{powers2018some}, and neural networks \cite{kunzel2018transfer}. 
Despite the vast literature on HTE estimation, evaluating the accuracy of a HTE estimator is in general open. 

There are two major motivations to study the assessment problem. 
First, an assessment approach measures the absolute performance of certain estimator on future data. 
Second, an assessment approach provides guidance for comparing estimators.
Aware that a large proportion of HTE estimators involve hyper-parameters, such as the amount of penalization in LASSO-based estimators, tree sizes in random-forests-based estimators, efficient model selection or tuning methods are ultra-important.

The major difficulty of the HTE assessment is attributed to the ``invisibility" of HTE. 
Standard assessment methods evaluate the performance of a predictor by comparing predictions to observations on a validation dataset.
The approach is valid since the observations are unbiased realizations of the values to be predicted. 
In contrast, in the potential outcome model, an observation is the response of a unit under treatment or control, whereas the HTE to be predicted is the difference of the two. 
Therefore, HTE is not observable and the standard assessment methods can not be applied. 

In this paper, we design a two-step assessment approach. 
In the first step, we match treated and control units and regard the differences in the responses of matched pairs as pseudo-observations of the HTE. 
In the second step, we compare predictions to the pseudo-observations and compute prediction error.
For matching, we propose a distance for a pair of treated unit and control unit based on proximity scores in random forests.
We also introduce a matching method which minimizes the average distance of pairs instead of the more-commonly-used total distance \cite{rosenbaum1991characterization}, and provide an algorithm adapted from the average minimum-cost flow problem.

For conducting the assessment approach with cross-validation, we recommend a match-then-split principle.
Explicitly, we first perform matching on the complete dataset, then split the matched pairs into different folds. 
Since the quality of matched pairs deteriorates as the sample size decreases, the pairs constructed by matching first consist of units more similar than those obtained by splitting first and matching within each fold. 
We remark that matching before splitting does not snoop the data thanks to the distance used.

The organization of the paper is as follows. 
In Section \ref{sec:background}, we introduce the background of the HTE assessment and discuss related works.
In Section \ref{sec:holdOut}, we introduce the assessment approach with a hold-out validation dataset.
In Section \ref{sec:crossValidation}, we discuss how to implement the assessment approach in the framework of cross-validation.
In Section \ref{sec:extension}, we extend the assessment approach to handle various types of responses.
In Section \ref{sec:simulation}, we compare several assessment approaches on synthetic data and data simulated from a real dataset.
In Section \ref{sec:discussion}, we discuss directions of future work.

\section{Background} \label{sec:background}

\subsection{Potential outcome model}
We consider the Neyman-Rubin potential outcome model with two treatment assignments, labeled as ``treatment" and ``control". 
We assume that there is an underlying population and observations are identically independent realizations. 
Explicitly, for unit $i$, there is a $p$ dimensional covariate vector $X_i$ sampled i.i.d. from an underlying distribution $\PP$. 
Given covariates $X_i$, a binary group assignment $W_i \in \{0,1\}$ is generated from the Bernoulli distribution with mean $e(X_i)$ (i.e. the propensity score). 
Unit $i$ is also associated with two potential outcomes $Y_i(0)$, $Y_i(1)$, where $Y_i(1)$ is observed if the unit is under treatment, and $Y_i(0)$ is observed if the unit is under control.
We assume the following models of potential outcomes
\begin{align*}
Y_i(1) | X_i &= \nu(X_i) + \varepsilon_i, \\ 
Y_i(0) | X_i &= \mu(X_i) + \varepsilon_i, 
\end{align*}
where $\nu(x)$ is the treatment group mean function, $\mu(x)$ is the control group mean function, $\varepsilon_i$ is some mean zero noise independent of $X_i$, $W_i$. 
We define HTE as the difference of group mean functions, that is $\tau(x) := \nu(x) - \mu(x)$. 
We summarize the data generation model as follows,
\begin{align}
\label{eq:model}
\begin{split}
    X_i &\stackrel{\text{iid}}{\sim} \PP, \\
    W_i | X_i &\stackrel{}{\sim} \text{Ber}(e(X_i)), \\
    Y_i | W_i, X_i &= \mu(X_i) + W_i \tau(X_i) + \varepsilon_i,
\end{split}
\end{align}

By considering model \eqref{eq:model}, we have implicitly made the following assumptions as in \cite{imbens2015causal}.

\begin{assumption}[Unconfoundedness]\label{assu:confound}
The assignment mechanism does not depend on potential outcomes:
\begin{align*}
    \left(Y_i^{(1)}, Y_i^{(0)}\right) \independent W_i~|~X_i.
\end{align*}
\end{assumption}

\begin{assumption}[Stable unit treatment value assumption]\label{assu:SUTVA}
The potential outcomes for any unit do not depend on the treatments assigned to other units. 
There are no different versions of each treatment level.
\end{assumption}


\subsection{Matching}
Assume that there are $n$ units in total: $n_t$ treated units $\{t_i\}_{1 \le i \le n_t}$ and $n_c$ control units $\{c_j\}_{1 \le j \le n_c}$
We define a match $\pi$ as a function from treated units to the subsets of control units. 
Let $\Pi$ be the associated set of matched pairs
\begin{align*}
    \Pi := \{(t_i, c_j): c_j \in \pi(t_i)\},
\end{align*}
and denote the number of pairs in set $\Pi$ as $|\Pi|$.
Note that there is a bijection between matches and sets of matched pairs, and we use two notations exchangeably. 
We define the multiplicity number of the treatment group in match $\pi$ as
\begin{align*}
    M_t^{\pi} := \max_{t_i} \sum_{c_j} \1_{\{c_j \in \pi(t_i)\}},
\end{align*}
and similarly we define $M_c^{\pi}$.
Let $d_{t_i, c_j}$ be a distance defined for each treatment-control pair $(t_i, c_j)$.
We denote the total distance and the average distance of a match $\pi$ under the distance $d_{t_i, c_j}$ by
\begin{align*}
    D_{\text{tot}}(\pi)
    := \sum_{t_i} \sum_{c_j \in \pi(t_i)} d_{t_i, c_j}, \quad
    D_{\text{ave}}(\pi)
    := \frac{D_{\text{tot}}(\pi)}{|\Pi|}.
\end{align*}

There is a fruitful literature of matching methods applied to causal inference problems. 
Generally, a matching method consists of two parts: matching distance and matching structure. 
Matching distance describes dissimilarity between a pair of units, such as covariate distance, propensity score difference.
Matching structure characterizes the skeleton of a match, such as pair matching, subset matching and full matching. See \cite{rosenbaum2019modern} and references therein for a detailed review of matching methods. 


\subsection{Related works}
In the literature of HTE, most works perform accuracy assessment by predicting the responses as follows.
On the training data, the treatment and control group mean functions are estimated, where the difference of two group mean functions are used as the HTE estimator. 
On the validation data, prediction errors of group mean functions are computed and used for assessing the accuracy of the corresponding HTE estimator. 
The issue of the method is that large prediction errors of group mean functions do not ruin out accurate HTE estimation, or
the estimators for mean group functions may of poor quality while the difference is still a reasonably good predictor of the HTE.
This may happen when the HTE enjoys better properties compared to the mean group functions, such as higher sparsity or smoothness \cite{Kunzel4156}. 
Moreover, if a HTE estimator comes without estimates for the mean group functions, predicting the response can not be carried out.

In \cite{athey2015machine}, an assessment method based on covariate matching is proposed. 
Each unit in the validation data is paired with a unit in the opposite treatment status and close with regard to covariates. 
Along this line, a pseudo-observation of HTE is obtained for each pair by taking the difference of the responses, and from here the standard prediction error computation can be applied. 
The method makes considerable progress in avoiding estimating the control group mean function, but is limited to the case where the dimension of covariates is not too large.

In \cite{athey2016recursive}, honest validation is proposed for causal recursive partitioning.  
Given a trained tree structure, honest validation compares the estimated values at each terminal node based on the training data and the validation data. 
The method cleverly utilizes the homogeneity of HTE at each terminal node, but it is not obvious how to generalize the method to other HTE estimators. 

We finally review two assessment methods for average treatment effect (ATE) estimations.
Synth-validation in \cite{schuler2017synth} generates synthetic data based on the observed data with a sequence of possible ATEs and evaluates the performance of ATE estimators by comparing them to the known effects. 
The approach can not be easily extended to HTE evaluation since the number of possible configurations of HTE increases exponentially with regard to the covariate dimension.
Another approach called within-study comparison in \cite{cook2008three} contrasts ATE estimators from observational studies with those from randomized experiments. 
The approach is not as effective for assessing HTE estimators due to the small sample size in each heterogeneity subgroup of randomized experiments.

\section{Assessment with hold-out validation dataset}\label{sec:holdOut}

\subsection{General framework}
In this section, we consider the HTE assessment with a hold-out validation dataset. 
We consider the following validation error for a HTE estimator $\hat{\tau}(x)$
\begin{align}
\label{eq:idealValidationError}
    \text{error}_{\text{ideal}}
    = \frac{1}{n_t}\sum_{t_i} \left(\tau_{t_i} - \hat{\tau}\left(X_{t_i}\right)\right)^2.
\end{align}
In the ideal world, for each treated unit, there is an identical copy that goes under control, and we replace $\tau_{t_i}$ in \eqref{eq:idealValidationError} by the differences in the responses.
In the real world, no identical copy exists, and thus we construct a match $\pi$ between treated units and control units. 
We regard the differences in responses as pseudo-observations of the HTE for the treated,
and estimate the ideal validation error in \eqref{eq:idealValidationError} by
\begin{align}
\label{eq:validationErrorEstimator}
    \widehat{\text{error}_{\pi}}
    = \frac{1}{|\Pi|}\sum_{(t_i, c_j) \in \Pi}\left(Y_{t_i} - Y_{c_j} - \hat{\tau}\left(X_{t_i}\right)\right)^2.
\end{align}

The proposition below characterizes the bias and variance of the validation error estimator $\widehat{\text{error}_{\pi}}$ conditioned on covariates and treatment assignments.
Define the oracle validation error for a match $\pi$ as 
\begin{align}
\label{eq:validationError}
    \text{error}_{\pi}
    = \frac{1}{|\Pi|}\sum_{(t_i, c_j) \in \Pi}\left(\tau_{t_i} - \hat{\tau}\left(X_{t_i}\right)\right)^2.
\end{align}
The oracle validation error and the validation error estimator are equal if the match $\pi$ is perfect and the potential outcomes are noiseless.
For a treated unit $t_i$ and a control unit $c_j$, define the difference in control group mean function values as $b_{t_i, c_j} = \mu(X_{t_i}) - \mu(X_{c_j})$. 
For a match $\pi$, define the mean squared differences in control group mean function values as $\overline{b^2_{\pi}} = \frac{1}{|\Pi|}\sum_{(t_i, c_j) \in \Pi} b_{t_i, c_j}^2$. 
\begin{proposition}
\label{prop:validationError}
Assuming model \eqref{eq:model}, $\var(\varepsilon) = \sigma^2$, $\var(\varepsilon^2) = \kappa \sigma^4$, we have
\begin{align*}
    \left(1 - \sqrt{\frac{\overline{b^2_{\pi}}}{\text{error}_{\pi}}}\right)^2
    &\le \frac{\EE\left[\widehat{\text{error}_{\pi}}\right]  - 2 \sigma^2}{\text{error}_{\pi}}
    \le \left(1 + \sqrt{\frac{\overline{b^2_{\pi}}}{\text{error}_{\pi}}}\right)^2,\\
    \var\left(\widehat{\text{error}_{\pi}}\right)
    &\le 
   \frac{M_{t}^{\pi} + M_{c}^{\pi} -1}{|\Pi|} \left((4\kappa +8) \sigma^4 + 32  \sigma^2 \left(\overline{b^2_{\pi}} + \text{error}_{\pi}\right)\right).
\end{align*}
\end{proposition}

Proposition \ref{prop:validationError} implies that a smaller $\overline{b^2_{\pi}}$ will result in smaller upper bounds for both the bias and variance of $\widehat{\text{error}_{\pi}}$. 
Besides, a larger $|\Pi|$ and smaller multiplicity numbers $M_t^{\pi}$, $M_c^{\pi}$ will induce a smaller upper bound for variance. 
We design a matching method based on the two observations.

\subsection{Matching distance}
Motivated by Proposition \ref{prop:validationError}, we match treated and control units with similar control group mean function values. 
The following steps are conducted on the validation dataset.
First, we learn the control group mean function via random forests using the control units.
Based on the random forest, we compute for each pair of treated unit and control unit a proximity score: the number of trees that the two units end up in the same terminal node.
We define the proximity score distance by subtracting the proximity score from the total number of trees. 
The proximity score distance is a pseudo-metric, and a pair of treated and control unit with small proximity score distance is close with regard to the control group mean function value in the eye of the random forest.

We compare the proximity score distance with other popular matching distances.
Propensity score distances are of little relevance here, because two units similar in the control group mean function value are not necessarily close in the propensity score, and vice versa.
Exact covariate matching is ideal but usually unrealistic.
Distances based solely on covariates usually treat covariates equally, and is inefficient when only a small proportion of the covariates are informative to the control group mean function.

Distances based on estimated control mean group functions serve for our goal, but rely more heavily on accurate estimates and are less robust to outliers.  
As Figure \ref{fig:proximityScoreCSEstimation} shows, matching on distances based on estimated control mean group functions may pair units far apart in the covariates influential to the control group mean function as in panel (a), while matching on the proximity score distance will result in pairs with close estimated control mean group functions as well as similar influential covariates as in panel (b). 
When the estimates are not accurate, pairs with similar covariates are more likely to stay close in the control group mean function value.
Besides, the proximity scores only depend on the tree structure, while the estimates also depend on the responses at each terminal node, and thus suffer more from outliers.

\begin{figure}[]
    \centering
    \begin{minipage}{7cm}
\includegraphics[scale = 0.37]{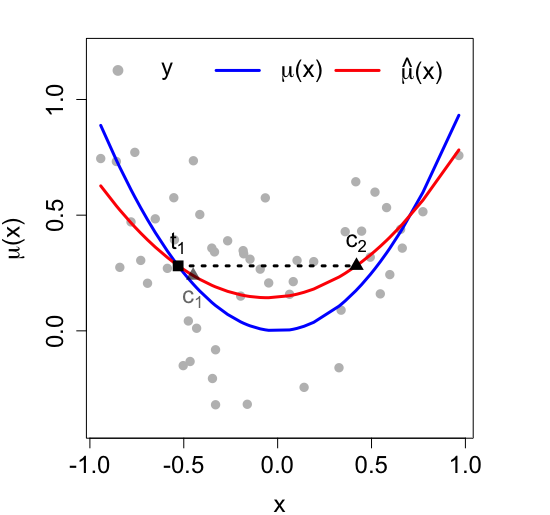}
    \subcaption*{\textit{\quad (a) pair $(t_1, c_2)$ favored by the distance 
    		\\\textcolor{white}{blank} based on the estimated control mean \\\textcolor{white}{blank} group function}}
    \end{minipage}
    \begin{minipage}{7cm}
\includegraphics[scale = 0.37]{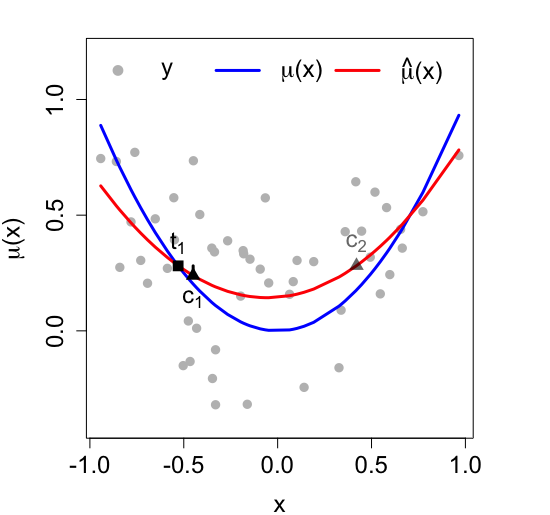}
    \subcaption*{\textit{\quad(b) pair $(t_1, c_1)$ favored by the proximity \\\textcolor{white}{blank} score distance \\\textcolor{white}{blank}}}
    \end{minipage}
    \caption{\textit{Comparison of the proximity score distance and the distance based on estimated control mean group functions.
    The blue curves are the true control group mean function, grey points are observations, and the red curves are the estimated control group mean function via least squares. 
    There are two candidate control units $c_1$ and $c_2$ to pair with the treated unit $t_1$, and $c_1$ is closer with regard to the control group mean function value, i.e. $\left|\mu(x_{t_1}) - \mu(x_{c_1})\right| < \left|\mu(x_{t_1}) - \mu(x_{c_2})\right|$.
    In the left panel, the distance based on the estimated control mean group function will prefer $c_2$ since $\hat{\mu}(x_{c_2}) = \hat{\mu}(x_{t_1})$.
    In the right panel, the proximity score distance will prefer $c_1$ since there is likely to be a split between $x_{t_1}$ and $x_{c_2}$, and thus $t_1$ and $c_2$ will end up in different terminal nodes leading to a large proximity score distance. 
    In the example, the proximity score distance finds the pair with closer control group mean function values.}
    }
    \label{fig:proximityScoreCSEstimation}
\end{figure}

\subsection{Matching structure}
Given a distance $d$ that captures the differences in the control group mean function values, by Proposition \ref{prop:validationError}, we aim to find a match in which $(1)$ paired control units and treated units are close regarding the provided distance; $(2)$ as many units as possible are used; $(3)$ no units are overused.

To illustrate, we consider the example in Figure \ref{fig:toyExample}. 
There are two equal-sized clusters $G_1$, $G_2$, where units in the same cluster share similar covariates and units not from the same clusters differ in covariates. 
As a result, control group mean function values are similar within clusters but different across clusters. 
Further assume that the units in cluster $G_2$ are more likely to be treated, and thus cluster $G_2$ has more treated units while cluster $G_1$ has more control units.
We remark that the example is motivated by the confounding phenomena in observational study: propensity score and baseline functions are influenced by the same covariates (i.e. confounders) known or unknown.
If the samples are clustered according to the confounder values, control group mean function values and proportions of treated units are different across clusters.

As depicted in Figure \ref{fig:toyExample},
there are three match candidates: 
in panel (a) each treated unit is matched to exactly one control unit and all the units are used, but there are undesirable matches across clusters; 
in panel (b) one-to-one matching is conducted and no pairs consist of units from different clusters, but part of the control units and treated units are not matched; 
in panel (c) there are no across-cluster pairs, every unit is matched, the treated units in cluster $G_1$ are used twice and similarly for the control units in cluster $G_2$.
Among the three matches, panel (c) satisfies the three properties aforementioned and is the most favorable candidate.

\tikzstyle{vertex}=[draw=black!25,shape=circle,fill=black!25, edge=black!25,minimum size=2pt,inner sep=0pt]
\tikzstyle{vertexTreated}=[draw=red!25,shape=circle,fill=red!75, edge=red!100, text = black, fontscale = 2, minimum size=20pt,inner sep=0pt]
\tikzstyle{vertexTreatedRemoved}=[draw=red!25,shape=circle,fill=red!75, edge=red!100, text = black, fontscale = 2, minimum size=20pt,inner sep=0pt, fill opacity = 0.5]
\tikzstyle{vertexControl}=[draw=blue!25,shape=circle,fill=blue!75, edge=blue!100, text = white, fontscale = 2, minimum size=20pt,inner sep=0pt]
\tikzstyle{vertexControlRemoved}=[draw=blue!25,shape=circle,fill=blue!75, edge=blue!100, text = white, fontscale = 2, minimum size=20pt,inner sep=0pt, fill opacity = 0.75]
\tikzstyle{vertexInvisible}=[draw=white!25,shape=circle,fill=white!75, edge=white!100, text = black, fontscale = 2, minimum size=20pt,inner sep=0pt, fill opacity = 1]
\tikzstyle{edge} = [draw,thick,-]
\tikzstyle{edgeRemove} = [draw,thick,-]
\tikzstyle{edgeAlert} = [draw,thick,-]
\tikzstyle{arrow} =[draw = darkgray,thick,-]
\tikzstyle{arrowRemove} = [draw = gray, thick,-, draw opacity = 0.8]
\tikzstyle{weight} = [font=\scriptsize]

\begin{figure}[]
\centering
\begin{tikzpicture}[scale=1.6,auto,swap]
    \foreach \pos /\name in {{(-0.6,0)}/G_1,{(-0.6,1.5)}/G_2}
        \node[vertexInvisible](\name) at \pos{$\name$};
    \foreach \pos /\name in {{(0,0)}/t_3,{(0,1)}/t_2,{(0,2)}/t_1}
        \node[vertexTreated](\name) at \pos{$\name$};
    \foreach \pos /\name in {{(1.5,0)}/c_3,{(1.5,1)}/c_2, {(1.5,2)}/c_1}
        \node[vertexControl](\name) at \pos{$\name$};
    \node[draw = white, fill = white, fill opacity = 0, edge = white, draw opacity = 0, text opacity = 1] at (0.4,-0.05) {};
    \node[draw = white, fill = white, fill opacity = 0, edge = white, draw opacity = 0, text opacity = 1] at (0.45,1) {};
    \node[draw = white, fill = white, fill opacity = 0, edge = white, draw opacity = 0, text opacity = 1] at (0.4,2.1) {};
    \foreach \source /\dest  in {t_1/c_1,t_3/c_3} 
        \path[arrow] (\source) -- node[weight] {} (\dest);
    \foreach \source /\dest  in {t_2/c_2} 
    \path[arrowRemove] (\source) -- node[weight] {} (\dest);
    \draw [rounded corners,fill opacity = 0, dashed, draw = black] (-0.3,-0.3)--(-0.3,0.2)--(1.8,1.6)--(1.8,-0.3)--cycle;
    \draw [rounded corners,fill opacity = 0, dashed, draw = black] (-0.3,2.3)--(-0.3,0.4)--(1.8,1.8)--(1.8,2.3)--cycle;
    \node[draw,text width=3cm, draw opacity = 0] at (0.75,-0.75) {\centering \textit{(a) undesirable}}; 
    \node[draw,text width=3cm, draw opacity = 0] at (1.15,-1) {\centering \textit{pair $(t_2, c_2)$}}; 
    \foreach \pos /\name in {{(3,0)}/t_3,{(3,2)}/t_1}
        \node[vertexTreated](\name) at \pos{$\name$};
    \foreach \pos /\name in {{(3,1)}/t_2}
        \node[vertexTreatedRemoved](\name) at \pos{$\name$};
    \foreach \pos /\name in {{(4.5,0)}/c_3, {(4.5,2)}/c_1}
        \node[vertexControl](\name) at \pos{$\name$};
    \foreach \pos /\name in {{(4.5,1)}/c_2}
        \node[vertexControlRemoved](\name) at \pos{$\name$};
    \node[draw = white, fill = white, fill opacity = 0, edge = white, draw opacity = 0, text opacity = 1] at (0.4,-0.05){};
    \node[draw = white, fill = white, fill opacity = 0, edge = white, draw opacity = 0, text opacity = 1] at (0.45,1){};
    \node[draw = white, fill = white, fill opacity = 0, edge = white, draw opacity = 0, text opacity = 1] at (0.4,2.1){};
    \foreach \source /\dest  in {t_1/c_1,t_3/c_3} 
        \path[arrow] (\source) -- node[weight] {} (\dest);
    \draw [rounded corners,fill opacity = 0, dashed, draw = black] (2.7,-0.3)--(2.7,0.2)--(4.8,1.6)--(4.8,-0.3)--cycle;
    \draw [rounded corners,fill opacity = 0, dashed, draw = black] (2.7,2.3)--(2.7,0.4)--(4.8,1.8)--(4.8,2.3)--cycle;
    \node[draw,text width=2.5cm,draw opacity = 0] at (2.25,2.75) {$M_t = M_c = 1$};
    \node[draw,text width=3.75cm, draw opacity = 0] at (4,-0.75) {\centering \textit{(b) unused units}}; 
    \node[draw,text width=3.75cm, draw opacity = 0] at (4.4,-1) {\centering \textit{$t_2$, $c_2$}}; 
    \foreach \pos /\name in {{(6,0)}/t_3,{(6,1)}/t_2,{(6,2)}/t_1}
        \node[vertexTreated](\name) at \pos{$\name$};
    \foreach \pos /\name in {{(7.5,0)}/c_3,{(7.5,1)}/c_2, {(7.5,2)}/c_1}
        \node[vertexControl](\name) at \pos{$\name$};
    \node[draw = white, fill = white, fill opacity = 0, edge = white, draw opacity = 0, text opacity = 1] at (0.4,-0.05){};
    \node[draw = white, fill = white, fill opacity = 0, edge = white, draw opacity = 0, text opacity = 1] at (0.45,1){};
    \node[draw = white, fill = white, fill opacity = 0, edge = white, draw opacity = 0, text opacity = 1] at (0.4,2.1){};
    \foreach \source /\dest  in {t_1/c_1,t_2/c_1,t_3/c_2,t_3/c_3} 
        \path[arrow] (\source) -- node[weight] {} (\dest);
    \draw [rounded corners,fill opacity = 0, dashed, draw = black] (5.7,-0.3)--(5.7,0.2)--(7.8,1.6)--(7.8,-0.3)--cycle;
    \draw [rounded corners,fill opacity = 0, dashed, draw = black] (5.7,2.3)--(5.7,0.4)--(7.8,1.8)--(7.8,2.3)--cycle;
    \draw [rounded corners,fill opacity = 0, draw = black] (5.25,-0.5)--(5.25,3)--cycle;
    \node[draw,text width=2.5cm,draw opacity = 0] at (6.75,2.75) {$M_t = M_c = 2$};
    \node[draw,text width=2cm,draw opacity = 0] at (6.75,-0.75) {\textit{\centering (c) desired}};
    \node[draw,text width=2cm,draw opacity = 0] at (7.15,-1) {\centering \textit{match}};
\end{tikzpicture}
\caption{\textit{Example for matching structure. 
There are two equal-sized clusters $G_1$, $G_2$, where units in the same cluster share similar covariates and units not from the same clusters differ in covariates. 
Control group mean function values are similar within clusters but different across clusters. Cluster $G_2$ has more treated units while cluster $G_1$ has more control units.
In (a), (b) $M_t= M_c = 1$, and in panel (c) $M_t = M_c = 2$.}}
\label{fig:toyExample}
\end{figure}
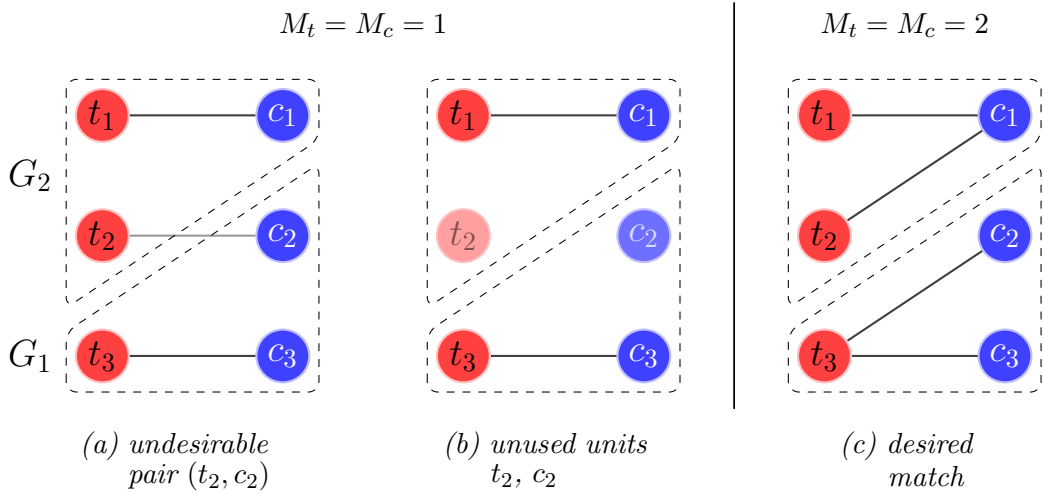

To find a match with the desired properties, such as panel (c) in the example in Figure \ref{fig:toyExample}, 
we propose the following matching objective 
\begin{align}
\label{eq:averageDistance}
    &\min_{\pi}~ D_{\text{ave}}(\pi)\\
\label{eq:groupControl}
    &m_c \le \sum_{t_i} \1_{\{c_j \in \pi(t_i)\}}  \le M_c, \quad \forall c_j,\\
\label{eq:groupTreatment}
    &m_t \le \sum_{c_j} \1_{\{c_j \in \pi(t_i)\}} \le M_t, \quad \forall t_i,
\end{align}
with pre-specified $m_c$, $m_t$, $M_c$, $M_t \ge 0$. 
The lower bounds in the multiplicity constraints \eqref{eq:groupControl},  \eqref{eq:groupTreatment} guarantee that as many units are used as possible.
The upper bounds in the multiplicity constraints \eqref{eq:groupControl}, \eqref{eq:groupTreatment} enforce that no units are matched excessively.
The objective function \eqref{eq:averageDistance}, focusing on the average distance, prefers a match with more good quality pairs to fewer poor quality pairs. 
Particularly for the example in Figure \ref{fig:toyExample}, the total distance minimization may rule out panel (c) since the total distance of many good quality pairs can be larger than that of fewer poor quality pairs, while the average distance always favors the former.
We discuss the multiplicity constraints \eqref{eq:groupControl}, \eqref{eq:groupTreatment} and the objective function \eqref{eq:averageDistance} in detail.

\subsubsection{Multiplicity constraints}
Arguably the most common multiplicity parameters are $M_t = M_c = 1$, and $m_t = 1$, $m_c = 0$. The constraint requires each treated unit be matched to one control unit and no control units are used multiple times. The constraint can be stringent if multiple control units are close to one treated unit and vice versa.
Consider the example in Figure \ref{fig:toyExample}.
If $M_t = M_c = 1$, $m_t = 1$ are enforced, a proportion of control units in cluster $G_1$ will be matched to treated units in cluster $G_2$ unfavorably as in panel (a).
If we relax $m_t = 1$ and avoid pairs across clusters, part of the control units in cluster $G_1$ and part of the treated units in cluster $G_2$ will not be matched as in panel (b), which reduces the efficiency of data usage. 
In contrast, consider $M_t = M_c = 2$, $m_t = m_c = 1$, where we allow treated units in cluster $G_1$ be matched to multiple control units, and the same for control units in cluster $G_2$. 
As depicted in panel (c), the match contains no pairs of units from different clusters, uses all the data and thus is more desirable compared to panel (a) and (b). 
In practice, we recommend $m_t = m_c = 1$ unless certain units are apparently outliers. 
For $M_t$, the matching method is more sensitive to small upper bounds than to large upper bounds, and thus we recommend to set $M_t$ reasonably large --- larger than the ratio of treated units over control units given any covariate value, and similarly for $M_c$.

\subsubsection{Objective function}
The objective function \eqref{eq:averageDistance} focuses on the average distance instead of the more commonly used total distance.
If the number of matched pairs is fixed, for instance, at the number of treated units, the total distance minimization and the average distance minimization are equivalent.
However, when the number of matched pairs is not fixed, the average distance minimization and the total distance minimization may favor different matches.

The following proposition further illustrates the differences between the average distance minimization and the total distance minimization.
\begin{proposition}
\label{prop:optimizationTarget}
If the optimization problem \eqref{eq:averageDistance} is feasible,
\begin{enumerate}
    \item the average distance minimization is invariant to the scale and the translation of distance, and the total distance minimization is invariant to the scale but not the translation of distance;
    \item given multiplicity parameters $M_t$, $M_c$, $m_t$, $m_c$ and let $\pi_{\text{ave}}$ and $\pi_{\text{tot}}$ denote an optimal solution of the average distance minimization and the total distance minimization respectively, then
    \begin{align*}
        D_{\text{ave}}(\pi_{\text{ave}})
        \le D_{\text{ave}}(\pi_{\text{tot}}),
        \quad \left|\Pi_{\text{ave}} \right|\ge \left|\Pi_{\text{tot}}\right|.
    \end{align*}
\end{enumerate}
\end{proposition}

To illustrate the importance of the translation invariance in Proposition \ref{prop:optimizationTarget}, we reconsider the previous example. 
As demonstrated in Figure \ref{fig:toyExample2}, we further assume that distances between units in the same cluster are $\Delta$, while those between units across clusters are $\Delta + \delta$.
One practical motivating distance is the semi-oracle distance $d_{t_i, c_j} = (Y_{t_i}(0) - Y_{c_j}(0))^2$, where
the expectation of the semi-oracle distance equals $2 \sigma^2$ for within-cluster pairs and $2 \sigma^2 + (\mu_2 - \mu_1)^2$ for across-cluster pairs. 
As the noise magnitude increases, the distance shifts up.
Another motivating distance is the covariate distance $d_{t_i, c_j} = \|X_{t_i} - X_{c_j}\|_2^2$. 
Suppose that the baseline function only depends on the first covariate and units are clustered according to $x_1$, then the distance is $ \sum_{k=2}^p(x_{k,t_i} - x_{k,c_j})^2$ for within-cluster pairs and $\sum_{k=2}^p(x_{k,t_i} - x_{k,c_j})^2 + (x_{1,t_i} - x_{1,c_j})^2$ for across-cluster pairs. 
As the dimension of covariates $p$ grows, the covariate distance is contaminated by the nuisance covariates. 

As demonstrated in Figure \ref{fig:toyExample2}, there are two match candidates : in panel $(a)$, there is one across-cluster pair, the total distance is $3 \Delta + \delta$ and the average distance is $\Delta + \delta/3$; in panel $(b)$, there is no across-cluster pair, the total distance is $4 \Delta$ and the average distance is $\Delta$. 
The average distance minimization always prefers the more favorable match with no across-cluster pairs in panel (b), while the total distance minimization prefers the match with unfavorable across-cluster pairs in panel (a) if $\Delta > \delta$. 
The translation invariance makes the average distance minimization robust to distance inflations.

To explain the benefit of (2) in Proposition \ref{prop:optimizationTarget}, if $D_{\text{ave}}(\pi)$ is relevant to $\overline{b_{\pi}^2}$, the average distance minimization reduces the bias and variance of the validation error estimator according to Proposition \ref{prop:validationError}.
Besides, a larger number of pairs constructed in the average distance minimization further reduces the variance of the validation error estimator. 

\begin{figure}[]
\centering
\begin{tikzpicture}[scale=1.6,auto,swap]
    \foreach \pos /\name in {{(-0.6,0)}/G_1,{(-0.6,1.5)}/G_2}
        \node[vertexInvisible](\name) at \pos{$\name$};
    \foreach \pos /\name in {{(0,0)}/t_3,{(0,1)}/t_2,{(0,2)}/t_1}
        \node[vertexTreated](\name) at \pos{$\name$};
    \foreach \pos /\name in {{(1.5,0)}/c_3,{(1.5,1)}/c_2, {(1.5,2)}/c_1}
        \node[vertexControl](\name) at \pos{$\name$};
    \node[draw = white, fill = white, fill opacity = 0, edge = white, draw opacity = 0, text opacity = 1] at (0.4,-0.05){};
    \node[draw = white, fill = white, fill opacity = 0, edge = white, draw opacity = 0, text opacity = 1] at (0.45,1){};
    \node[draw = white, fill = white, fill opacity = 0, edge = white, draw opacity = 0, text opacity = 1] at (0.4,2.1){};
    \node[draw = white, fill = white, fill opacity = 0, edge = white, draw opacity = 0, text opacity = 1] at (0.75,2.2) {$\Delta$};
    \node[draw = white, fill = white, fill opacity = 0, edge = white, draw opacity = 0, text opacity = 1] at (0.75,1.2) {$\Delta+\delta$};
    \node[draw = white, fill = white, fill opacity = 0, edge = white, draw opacity = 0, text opacity = 1] at (0.75,0.2) {$\Delta$};
    \foreach \source /\dest  in {t_1/c_1,t_3/c_3} 
        \path[arrow] (\source) -- node[weight] {} (\dest);
    \foreach \source /\dest  in {t_2/c_2} 
    \path[arrowRemove] (\source) -- node[weight] {} (\dest);
    \draw [rounded corners,fill opacity = 0, dashed, draw = black] (-0.3,-0.3)--(-0.3,0.2)--(1.8,1.6)--(1.8,-0.3)--cycle;
    \draw [rounded corners,fill opacity = 0, dashed, draw = black] (-0.3,2.3)--(-0.3,0.4)--(1.8,1.8)--(1.8,2.3)--cycle;
    \node[draw,text width=3cm, draw opacity = 0] at (0.75,-0.75) {\textit{(a) across-cluster}};
    \node[draw,text width=3cm, draw opacity = 0] at (1.15,-1) {\textit{pair $(t_2, c_2)$}};
    \foreach \pos /\name in {{(3,0)}/t_3,{(3,1)}/t_2,{(3,2)}/t_1}
        \node[vertexTreated](\name) at \pos{$\name$};
    \foreach \pos /\name in {{(4.5,0)}/c_3,{(4.5,1)}/c_2, {(4.5,2)}/c_1}
        \node[vertexControl](\name) at \pos{$\name$};
    \node[draw = white, fill = white, fill opacity = 0, edge = white, draw opacity = 0, text opacity = 1] at (0.4,-0.05){};
    \node[draw = white, fill = white, fill opacity = 0, edge = white, draw opacity = 0, text opacity = 1] at (0.45,1){};
    \node[draw = white, fill = white, fill opacity = 0, edge = white, draw opacity = 0, text opacity = 1] at (0.4,2.1){};
    \node[draw = white, fill = white, fill opacity = 0, edge = white, draw opacity = 0, text opacity = 1] at (3.75,2.2) {$\Delta$};
    \node[draw = white, fill = white, fill opacity = 0, edge = white, draw opacity = 0, text opacity = 1] at (3.75,1.3) {$\Delta$};
    \node[draw = white, fill = white, fill opacity = 0, edge = white, draw opacity = 0, text opacity = 1] at (3.75,0.7) {$\Delta$};
    \node[draw = white, fill = white, fill opacity = 0, edge = white, draw opacity = 0, text opacity = 1] at (3.75,0.2) {$\Delta$};
    \foreach \source /\dest  in {t_1/c_1,t_2/c_1,t_3/c_2,t_3/c_3} 
        \path[arrow] (\source) -- node[weight] {} (\dest);
    \draw [rounded corners,fill opacity = 0, dashed, draw = black] (2.7,-0.3)--(2.7,0.2)--(4.8,1.6)--(4.8,-0.3)--cycle;
    \draw [rounded corners,fill opacity = 0, dashed, draw = black] (2.7,2.3)--(2.7,0.4)--(4.8,1.8)--(4.8,2.3)--cycle;
    \node[draw,text width=3cm,draw opacity = 0] at (3.75,-0.75) {\textit{(b) no across-}}; 
    \node[draw,text width=3cm,draw opacity = 0] at (4.15,-1) {\textit{cluster pairs}}; 
\end{tikzpicture}
\caption{\textit{Comparison of the average distance minimization and the total distance minimization (continued from the example in Figure \ref{fig:toyExample}). 
Distances between units in the same cluster and across clusters are $\Delta$ and $\Delta + \delta$ respectively.
In panel $(a)$, there is one across-cluster pair, the total distance is $3 \Delta + \delta$ and the average distance is $\Delta + \delta/3$; in panel $(b)$, there is no across-cluster pair, the total distance is $4 \Delta$ and the average distance is $\Delta$. 
The average distance minimization always prefers the match in panel (b), while the total distance minimization prefers the match in panel (a) if $\Delta > \delta$. }}
\label{fig:toyExample2}
\end{figure}
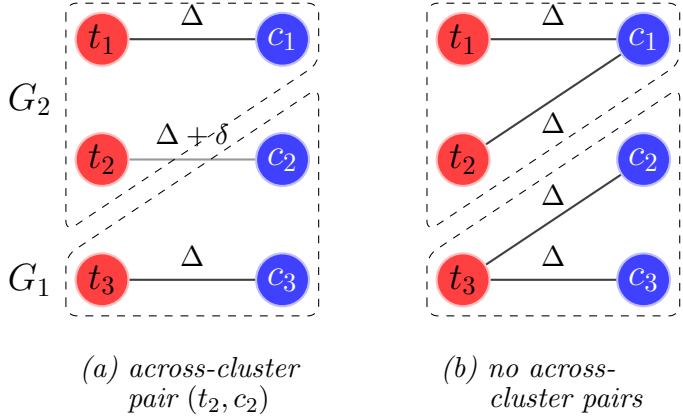

\subsubsection{Computation}
In general, there are two major approaches to solve a matching problem. 
The first approach casts the matching problem as linear programming, then applies extensive optimization tools therein. 
The objective function of the total distance minimization is linear, and is approachable via linear programming. 
However, the objective function of the average distance minimization is non-linear, thus algorithms for linear programming can not be directly applied.
The second approach formulates the matching problem as a minimum-cost flow problem \cite{rosenbaum1989optimal}. 
Standard minimum-cost flow problem requires to input the flow value, or equivalently the total number of pairs in the match. 
Unfortunately, the flow value is not directly available in the average distance minimization.

We propose an algorithm for the average distance minimization derived from \cite{chen1995minimal}. Explicitly, we search for the optimal flow value via binary search, and in each sub-routine we solve a minimum-cost flow problem. 
The algorithm is of the same time complexity as solving one minimum-cost flow problem up to logarithmic factor of the maximal number of allowed pairs.
Typically, the pair matching in \cite{rosenbaum1989optimal}, which minimizes the total distance and enforces each treated units to be matched exactly once, is of time complexity $O(n^3)$, and the average distance minimization takes $O\left(n^3 \log(n(M_t + M_c)\right)$.

Finally, we summarize the assessment approach with a hold-out validation dataset in Algorithm \ref{algo:holdOut}.
\begin{algorithm}[]
	 \DontPrintSemicolon  
	 \SetAlgoLined
	  \BlankLine
	\caption{Assessment approach with hold-out dataset	\label{algo:holdOut}}
	\textbf{Input:} HTE estimator $\hat{\tau}$, validation set $\{(X_i,W_i,Y_i)\}$, multiplicity parameters $M_t$, $M_c$, $m_t$, $m_c$. \\
	(1) Build a random forest $\{T_l\}_{1\le l \le m}$ with $m$ trees on the control group of the validation dataset. 
	Compute the proximity score distance for each pair of treated unit $t_i$ and control unit $c_j$ as 
	\begin{align*}
	    d_{t_i, c_j} = \sum_{l=1}^m \1_{\left\{T_l(X_{t_i}) \neq T_l(X_{c_j})\right\}},
	\end{align*} 
	where $T_l(x)$ denotes the terminal node of tree $T_l$ that a unit with covariate $x$ falls into.
	\\
	(2) Solve the average distance minimization problem with the distance $d_{t_i, c_j}$, multiplicity constraints $M_t$, $M_c$, $m_t$, $m_c$ and obtain match $\pi$.
	\\
	(3) Compute the validation error estimator of match $\pi$ in Eq.\eqref{eq:validationErrorEstimator} 
	\begin{align*}  
	   \widehat{\text{error}_\pi} = \frac{1}{|\Pi|}\sum_{(t_i, c_j) \in \Pi} \left(Y_{t_i} - Y_{c_j} - \hat{\tau}\left(X_{t_i}\right)\right)^2
	\end{align*}
	and output $\widehat{\text{error}_\pi}$.
\end{algorithm}

\section{Assessment with cross-validation}\label{sec:crossValidation}

In practice, hold-out datasets may be costly. 
A popular validation paradigm that uses the whole dataset for training while providing a reasonably good evaluation of the estimation performance is cross-validation. 
In this section we discuss how to conduct the assessment approach under the framework of cross-validation.

The standard cross-validation consists of two steps: first, split the data into several folds randomly equally; second, train on all but one fold, conduct validation on the left-out fold, and repeat this for each fold.
Naively integrating the assessment approach and the standard cross-validation raises the issue: the former splitting hurts the later matching.
Consider the most favorable case where there are natural pair-structures in samples, by splitting first, we may assign two naturally paired units to different folds, and thus miss the perfect match.

\begin{figure}[]
\centering
\begin{tikzpicture}[scale=1.6,auto,swap]
    \foreach \pos /\name in {{(0,0)}/t_3,{(0,1)}/t_2,{(0,2)}/t_1}
        \node[vertexTreated](\name) at \pos{$\name$};
    \foreach \pos /\name in {{(1.5,0)}/c_3,{(1.5,1)}/c_2, {(1.5,2)}/c_1}
        \node[vertexControl](\name) at \pos{$\name$};
    \node[draw = white, fill = white, fill opacity = 0, edge = white, draw opacity = 0, text opacity = 1] at (0.4,-0.05){};
    \node[draw = white, fill = white, fill opacity = 0, edge = white, draw opacity = 0, text opacity = 1] at (0.75,0.63) {$1$};
    \node[draw = white, fill = white, fill opacity = 0, edge = white, draw opacity = 0, text opacity = 1] at (0.75,1.63) {$1$};
    \node[draw = white, fill = white, fill opacity = 0, edge = white, draw opacity = 0, text opacity = 1] at (0.75,2.13) {$4$};
    \node[draw = white, fill = white, fill opacity = 0, edge = white, draw opacity = 0, text opacity = 1] at (0.75,1.13) {$2$};
    \node[draw = white, fill = white, fill opacity = 0, edge = white, draw opacity = 0, text opacity = 1] at (0.75,0.13) {$4$};
    \foreach \source /\dest  in {t_1/c_1,c_1/t_2, t_2/c_2, c_2/t_3, t_3/c_3}
        \path[arrow] (\source) -- node[weight] {} (\dest);
    \node[draw,text width=3.5cm, draw opacity = 0] at (0.75,-0.75) {\textit{(a) before pruning}};
    \foreach \pos /\name in {{(3,0)}/t_3,{(3,1)}/t_2,{(3,2)}/t_1}
        \node[vertexTreated](\name) at \pos{$\name$};
    \foreach \pos /\name in {{(4.5,0)}/c_3,{(4.5,1)}/c_2, {(4.5,2)}/c_1}
        \node[vertexControl](\name) at \pos{$\name$};
    \node[draw = white, fill = white, fill opacity = 0, edge = white, draw opacity = 0, text opacity = 1] at (0.4,-0.05){};
    \node[draw = white, fill = white, fill opacity = 0, edge = white, draw opacity = 0, text opacity = 1] at (0.45,1){};
    \node[draw = white, fill = white, fill opacity = 0, edge = white, draw opacity = 0, text opacity = 1] at (3.75,1.63){};
    \node[draw = white, fill = white, fill opacity = 0, edge = white, draw opacity = 0, text opacity = 1] at (3.75,2.13) {$4$};
    \node[draw = white, fill = white, fill opacity = 0, edge = white, draw opacity = 0, text opacity = 1] at (3.75,1.63) {$1$};
    \node[draw = white, fill = white, fill opacity = 0, edge = white, draw opacity = 0, text opacity = 1] at (3.75,0.63) {$1$};
    \node[draw = white, fill = white, fill opacity = 0, edge = white, draw opacity = 0, text opacity = 1] at (3.75,0.13) {$4$};
    \foreach \source /\dest  in {t_1/c_1,t_2/c_1,t_3/c_2,t_3/c_3} 
        \path[arrow] (\source) -- node[weight] {} (\dest);
    \foreach \source /\dest  in {t_2/c_2} 
        \path[arrowRemove] (\source) -- node[weight] {} (\dest);
    \draw [rounded corners,fill opacity = 0, dashed, draw = black] (3.5,0.7)--(4,1.3)--cycle;
    \node[draw,text width=3.5 cm,draw opacity = 0] at (3.75,-0.75) {\textit{(b) after pruning}}; 
\end{tikzpicture}
\caption{\textit{Example of pruning. 
In panel (a), the graph forms a chain of treated units and control units alternately.
There are three removable edges $(t_2, c_1)$, $(t_2, c_2)$, $(t_3, c_2)$. 
We pick the removable edge with the maximal distance, i.e. $(t_2, c_2)$, eliminate the edge and obtain panel (b). 
After pruning $(t_2, c_2)$, edges $(t_2, c_1)$ and $(t_3, c_2)$ are no longer removable, and the set of removable edges is empty. 
Therefore, we stop pruning.
In the pruned match, units can be split into two connected subgroups: $\{t_1, t_2, c_1\}$ and $\{t_3, c_2, c_3\}$. 
The connected subgroups either consist of one treated unit and multiple control units, or vice versa.}}
\label{fig:chain}
\end{figure}
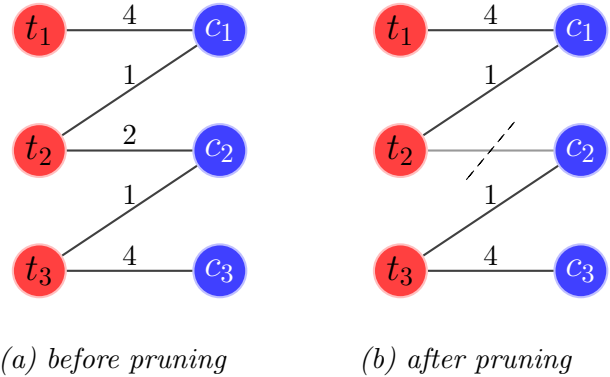

To tackle this problem, we propose to do matching prior to splitting, short as match-then-split.
Particularly, on the whole dataset, we obtain proximity score distances and solve the average distance minimization to obtain the optimal match.
We next split the samples into folds preserving the pair-structures, in other words, we avoid assigning matched units to different folds.
Applying the match-then-split principle to the aforementioned example with perfect pairs, we first match each unit with its identical copy, and then randomly split the pairs into folds without breaking the matched units apart.

A natural concern of the match-then-split principle is data snooping. 
However, notice that the distance metric for matching is obtained solely on the control group data and the treatment group is not touched, the one-sided data provides no information for the differences between the two sides. 
Therefore, splitting after matching is blind to the validation target and valid. 

A difficulty arises for data splitting in order to keep matched units in the same fold. 
We represent a match by a undirected graph where each node represents a unit, and there is an edge between two nodes if and only if the two units are matched. 
The pair-preserving constraint implies that connected components should stay together.
Since each unit is allowed to be matched multiple times, there may exist large connected components as depicted in panel (a) of Figure \ref{fig:chain}.
In the extremist scenario, the graph may be connected itself, and splitting without breaking pairs is impossible.

To enable proper splitting, we modify the average distance minimization.
Beyond the multiplicity constraints \eqref{eq:groupControl}, \eqref{eq:groupTreatment}, we further restrict the maximal path length of the graph to be at most three.
As a result, there are only two possible types of connected components: (1) one treated unit with multiple matched control units; (2) one control unit with multiple matched treated units.
The maximal size of the connected components are upper bounded by $1 + \max\{M_t, M_c\}$, which is usually small.
In this way, we can assign the connected components randomly into folds without destroying pair-structures.
We remark that the extra constraint is also adopted in full matching \cite{rosenbaum1991characterization}. 

The new constraint poses an extra challenge to computation.
In full matching where the total distance is minimized, the constraint is automatically fullfilled. 
However, this is not true for the average distance minimization. 
In fact, no known efficient network algorithm works under the path length constraint. 
As a surrogate, we propose the following heuristic pruning algorithm. 
Particularly, we start with the solution of the average distance minimization. 
We call an edge $(t_i, c_j)$ removable if the treated unit $t_i$ is matched to more than one control units, and the control unit $c_j$ is matched to more than one treated units. 
The new constraint is equivalent to the condition that there are no removable edges in the graph. 
We iteratively prune the highest cost removable edge until the set of removable edges is empty. 
See Figure \ref{fig:chain} for an example.
The algorithm is summarized in Algorithm \ref{algo:prune}. 

\begin{algorithm}[]
\caption{Pruning}
\label{algo:prune}
\begin{algorithmic}
\State Input: distance $d$, match $\Pi$.
\State Find the set of removable edges, i.e. the edges whose vertices are both connected to more than one vertex,
\begin{align}
\label{eq:removable}
    \calA ：= \left\{e_{t_i, c_j} \in \Pi: \sum_{c_k} \1_{\{(t_i,c_k) \in \Pi\}} \ge 2, ~\sum_{t_k} \1_{\{(t_k,c_j) \in \Pi\}} \ge 2 \right\}.
\end{align}
\While {$A \neq \emptyset$}
{\State (1) Find the removable edge with the maximal distance
\begin{align*}
    e_{t_i, c_j}^- = \argmax_{e_{t_k, c_l} \in \calA} d_{t_k, c_l}. 
\end{align*}
\State (2) Prune edge $e_{t_i, c_j}^-$: $\Pi \leftarrow \Pi \slash \{e_{t_i, c_j}^-\}$. 
\State (3) Update the set of removable edges $\calA$.}
\State Output the pruned match $\Pi$.
\end{algorithmic}
\end{algorithm}

We discuss properties of the pruned match.
First, pairs after pruning are a subset of the set of matched pairs from the average distance minimization, thus multiplicity constrains \eqref{eq:groupControl}, \eqref{eq:groupTreatment} are satisfied. 
Second, if the match without the path-length constraint is able to avoid low-quality pairs, the pruned match will automatically keep away from those pairs by choosing from existed pairs.
Third, by eliminating the removable pair with the maximal distance each time, we are heading towards the optimal solution with the path-length constraint greedily.

\section{Simulation} \label{sec:simulation}
\subsection{Simulating from model}
In this section, we compare various validation methods under the cross-validation framework on the synthetic data generated from model \eqref{eq:model}.
We vary four ingredients of a validation method:
\begin{enumerate}
    \item \textit{Target of comparison}. 
    We consider two targets of comparisons:
    (1) we obtain estimators of the HTE and the control group mean function (or equivalently estimators of the treatment and  control group mean functions) in training, and compare the estimators to the responses in validation;
    (2) we match treated units and control units in validation and compare the HTE estimators obtained in training to the differences between the responses of the matched pairs.
    \item \textit{Matching distance}. 
    We compare the proximity score distance, and the Mahalanobis distance of covariates $(X_{t_i} - X_{c_j})^\top \Sigma^{-1} (X_{t_i} - X_{c_j})$, where $\Sigma$ denotes the covariance matrix of the covariates. 
    \item \textit{Matching structure}. 
    We compare the average distance minimization with pruning and the total distance minimization. 
    The total distance minimization is available in the R package \textit{optmatch}.
    \item \textit{Split or match first}. We compare the match-then-split and the split-then-match discussed in Section \ref{sec:crossValidation}. 
\end{enumerate}
Based on the four ingredients, we consider the five validation methods in Table \ref{tab:method}. 
\begin{table}[]
\centering
    \begin{tabular}{c|cccc}
    \hline
       method &  target of &  matching  &   matching &  split or \\
       abbreviation &  comparison &  distance &  structure &  match first\\\hline
      prd &  $\quad$ response $\quad$ &   - &  - &   -\\ 
     cvr &   HTE &   covariate dist. &  $~$ average dist. $~$ &   $\qquad$ match $\qquad$\\ 
     full &   HTE &   prox. score dist. &   total dist. &   match \\
      S-M &   HTE &   prox. score dist. &   average dist. &   split\\
      $\qquad$ combo $\qquad$ &   HTE &   prox. score dist. &   average dist. &   match\\ 
     \hline
    \end{tabular}
    \caption{\textit{Summary of the features of validation methods. The prd method compares the estimators of treatment or control group mean function to the responses. The other four methods contrast the HTE estimator with the differences between the responses of matched pairs. Explicitly, the cvr method considers the Mahalanobis distance of covariates; the full method optimizes the total distance; the S-M method first splits data, then constructs pairs within each fold separately; the combo method considers the proximity score distance, minimizes the average distance and obeys the match-then-split principle. }}
    \label{tab:method}
\end{table}

As for data generation, we consider linear HTE $\tau(x) = x^\top\beta$ where $x$ includes the intercept.
We vary four critical factors affecting the performance of the aforementioned validation methods:
\begin{enumerate}
    \item \textit{Control group mean function}. We consider the control group mean function $\mu(x) =  x^\top\alpha + \delta \cdot |x_1|$. When $\delta \neq 0$, $\mu(x)$ is not linear in $x$. 
    \item \textit{Dimension of covariates}. We set the dimension of covariates $p \in \{10, 20\}$. 
    \item \textit{Propensity score}. We consider the constant propensity score $e(x) = 0.5$, and the covariate-dependent propensity score $e(x) = \frac{e^{x^\top \theta}}{1 + e^{x^\top\theta}}$. In particular, we set $e(x)$ and $\tau(x)$ to be positively correlated, which agrees with the fact that the units benefit more from the treatment are more likely to be treated. 
    \item \textit{Number of folds}. We set the number of folds $k \in \{10,25\}$.
\end{enumerate}
Based on the four factors, we consider the following five simulation settings in Table \ref{tab:scenario}.
Moreover, without further specification we consider the sample size $n = 200$, and covariates drawn i.i.d. uniformly on $[-1,1]$.
For the HTE and the linear part of the control group mean function, we restrict at least half of the coefficients to be zero. 
We control the signal noise ratio $\var((W-e(X))\tau(X))/\var(\varepsilon)$ to be less than $1$. 
We repeat each setting $200$ times and aggregate the results. 

\begin{table}[H]
    \centering
    \begin{tabular}{c|cccc}
    \hline
     simulation & nonlinear-  &dimension of    & propensity & number of \\
     setting & arity $(\delta)$  & covariates $(p)$ &  score $(e(x))$ &folds $(k)$\\\hline
    \RNum{1} & $0$ & $10$ & $0.5$ & $10$ \\ 
    \RNum{2} & $-2$ & $10$ & $0.5$ & $10$ \\ 
    \RNum{3} & $0$ & $20$ & $0.5$ & $10$ \\ 
    \RNum{4} & $0$ & $10$ & $\frac{e^{2x_1}}{1+e^{2x_1}}$ & $10$ \\
    \RNum{5} & $0$ & $10$ & $0.5$ & $25$  \\
    \hline
    \end{tabular}
    \caption{\textit{Summary of the features of simulation settings. Setting \RNum{1} is the default setting; setting \RNum{2} considers non-linear control mean group function; in setting \RNum{3}, the number of covariates is larger; in setting \RNum{4}, the propensity score and the treatment effect are positively correlated; in setting \RNum{5}, data are split into more folds in cross-validation.}}
    \label{tab:scenario}
\end{table}

As for the HTE estimator, we consider the following LASSO-based approach 
\begin{align}
\label{eq:LASSO}
    (\hat{\alpha},\hat{\beta})
    = \argmin_{\alpha, \beta}\frac{1}{2n} \sum_{i=1}^n (Y_i -  X_i\alpha - W_i \cdot X_i \beta)^2 + \lambda \left(\|\alpha\|_1 + \|\beta\|_1 \right),
\end{align}
with a sequence of tuning parameters $\lambda$. 
The approach is a starting point of the HTE estimation with variable selection, which works under the simple linear model, and involves only one tuning parameter. 
We expect that a good validation method should at least work well with the simple estimation approach.

As for comparison criteria, we evaluate the tuning performance of validation methods.
In particular, for each setting in Table \ref{tab:scenario}, we run each validation method in Table \ref{tab:method} under the cross-validation framework and pick the tuning parameter $\lambda_{\text{method}}$ of the minimal validation error. 
We then solve \eqref{eq:LASSO} on the whole dataset with the tuning parameter $\lambda_{\text{method}}$ and obtain estimator $\hat{\beta}_{\lambda_{\text{method}}}$. 
We denote the estimation error $\|\hat{\beta}_{\lambda_{\text{method}}} - \beta\|_2^2$ by $\text{MSE}_{\text{method}}$. 
Meanwhile, we define oracle estimation error as $\text{MSE}_{\text{oracle}} = \min_{\lambda} \|\hat{\beta}_{\lambda} - \beta\|_2^2$.
The tuning performance of a validation method is assessed by the log ratio of $\text{MSE}_{\text{method}}$ over $\text{MSE}_{\text{oracle}}$, referred to as relative MSE in the following,
\begin{align}
\label{eq:relativeMSE}
    \log\left(\frac{\text{MSE}_{\text{method}}}{\text{MSE}_{\text{oracle}}}\right).
\end{align}

We also compare the shape of validation curves.
In each trial, for the sequence of tuning parameters, we compute the oracle estimation error and validation errors. 
We then average the errors over trials and obtain error curves.
By Proposition \ref{prop:validationError},  validation error curves should be similar to the oracle error curve up to shift. 
In other words, a favorable validation error curve should be parallel to the oracle error curve, but not necessarily coincide.
To evaluate the degree of parallel, we regress validation error curves over the oracle error curve. 
We present regression coefficients, which is ideally one. 
Note that a close-to-one regression coefficient does not imply the validation error curve is similar to the oracle error curve, therefore we also present the $R^2$ of the regressions, which are proportional to the correlation between the oracle error curve and validation error curves. 

According to the simulation results in Figure \ref{fig:simulation} and Table \ref{tab:simulation},
the \textit{combo} method: a combination of the proximity score distance, the average distance minimization with pruning and the match-then-split principle, performs favorably. 
The method selects the tuning parameter corresponding to the lowest relative MSE, and produces the validation error curve the most similar to the oracle. 
More specific comparisons are discussed to investigate how the four ingredients of a validation method matter.
\begin{itemize}
     \item \textit{Target of comparison and model misspecification}. 
     Comparing setting \RNum{1} and \RNum{2}, when the control group mean function is misspecified, the \textit{prd} method comparing estimators with responses performs worse.
     The reason of adding nonlinear terms into the control group mean function instead of the HTE is as follows.
     According to domain knowledge, the control group mean function, e.g. blood pressure, is usually influenced by more factors than the HTE, e.g. the difference in blood pressure induced by a therapy, and in a more complicated way. Moreover, the HTE can be interpreted as the interaction between the treatment assignment and covariates. 
     A common hierarchical assumption of interaction is that a covariate does not go into interaction if it does not appear in the main effect. 
     \item \textit{Dimension of covariate and matching distance}. 
     Compare setting \RNum{1} and \RNum{3}, as there are more irrelevant predictors, the \textit{cvr} method is less favorable since the quality of the covariate distance deteriorates while the proximity score distance remains informative.
     \item \textit{Proximity score and matching structure}. 
     In the presence of confounding, setting \RNum{4} is similar to the example in Figure \ref{fig:toyExample2}, and the total distance minimization in the \textit{full} method performs relatively unsatisfactory. 
     \item \textit{Number of fold and splitting or matching first}. 
     Compare setting \RNum{1} and \RNum{5}, as the number of folds grows, the quality of pairs decays and the \textit{S-M} method is less attractive.
     In the extreme case where each fold is of size two: one treated unit and one control unit, an analogy to the leave-one-out cross validation, there is essentially no matching. 
 \end{itemize}
The code of the proximity score distance construction and the average distance minimization will soon be available on \textit{github}.

\begin{figure}[]
    \centering
    \begin{minipage}{5cm}
    \centering  
\includegraphics[scale = 0.23]{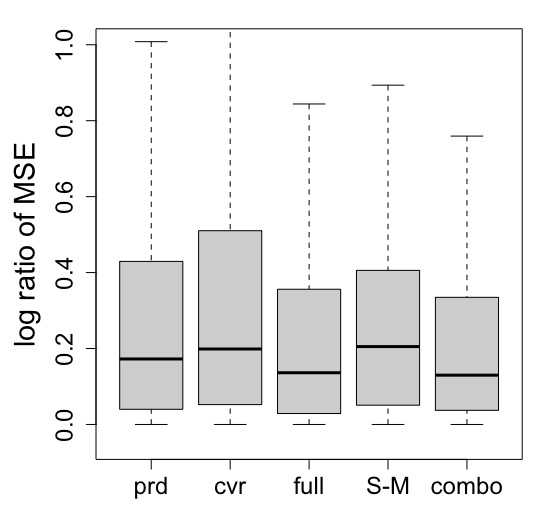}
    \subcaption*{\textit{(a) setting \RNum{1}}}
    \end{minipage}
    \begin{minipage}{5cm}
    \centering  
\includegraphics[scale = 0.23]{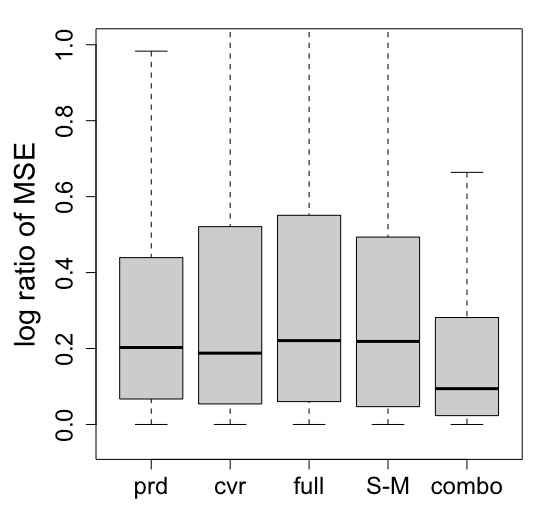}  
    \subcaption*{\textit{(b) setting \RNum{2}}}
    \end{minipage}
    \begin{minipage}{5cm}
    \centering  
\includegraphics[scale = 0.23]{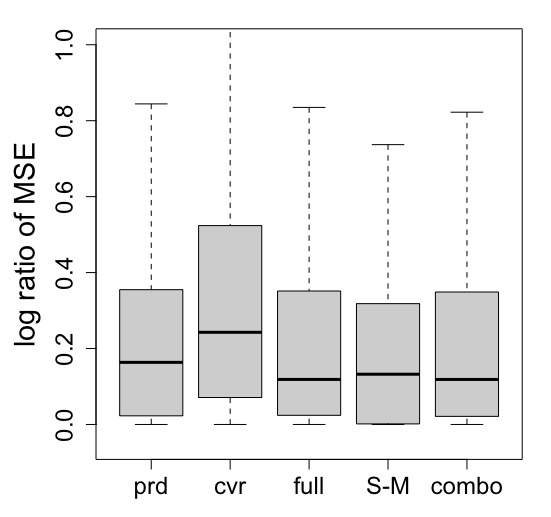}
    \subcaption*{\textit{(c) setting \RNum{3}}}
    \end{minipage}
    \begin{minipage}{5cm}
    \centering  
\includegraphics[scale = 0.23]{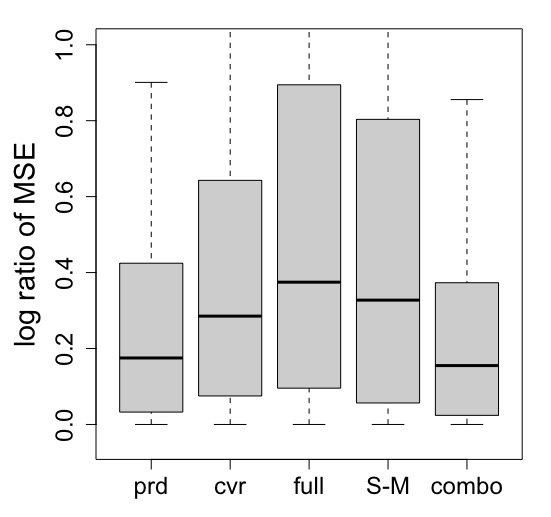}  
    \subcaption*{\textit{(d) setting \RNum{4}}}
    \end{minipage}
    \begin{minipage}{5cm}
    \centering  
\includegraphics[scale = 0.23]{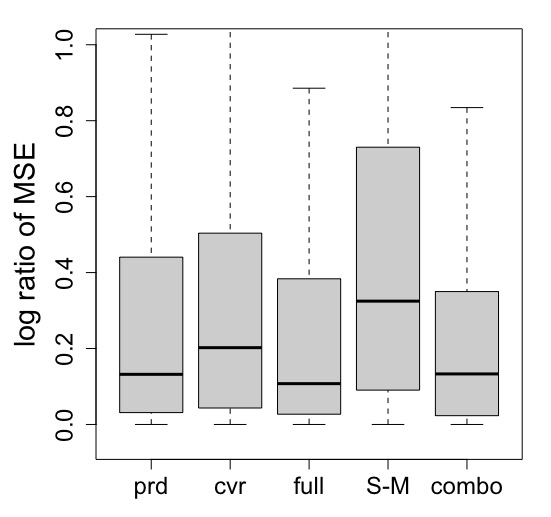}  
    \subcaption*{\textit{(e) setting \RNum{5}}}
    \end{minipage}
    \caption{\textit{Comparison of the tuning performance. 
    We display the boxplots of the log ratio of $\text{MSE}_{\text{method}}$ over $\text{MSE}_{\text{oracle}}$ of the five validation methods in Table \ref{tab:method} under the five simulation settings in Table \ref{tab:scenario}. 
    Each setting is repeated $200$ times. }}
    \label{fig:simulation}
\end{figure}

\begin{table}[]
	\centering
	\begin{tabular}{c|cccccccccc}
		\hline
		\multirow{2}{*}{method} & \multicolumn{2}{c}{I} & \multicolumn{2}{c}{II} & \multicolumn{2}{c}{III} & \multicolumn{2}{c}{IV} & \multicolumn{2}{c}{V} \\ \cline{2-11} 
		& \quad coef.  \quad   & $R^2$\qquad    & \quad coef.  \quad     & $R^2$\quad        & \quad coef.  \quad    & $R^2$\quad         & \quad coef.  \quad      & $R^2$\quad        & \quad coef.  \quad & $R^2$\quad       \\ \hline
		prd                     & $2.77$    & $0.94$    & $1.93$     & $0.80$    & $1.98$     & $0.59$     & $3.30$     & $0.93$    & $2.63$    & $0.96$    \\
		cvr                     & $1.30$    & $0.84$    & $0.24$     & $0.03$    & $0.33$     & $0.15$     & $1.30$     & $0.50$    & $1.29$    & $0.85$    \\
		full                    & $0.89$    & $1.00$    & $0.15$     & $0.01$    & $0.92$     & $0.98$     & $1.27$     & $0.35$    & $0.85$    & $1.00$    \\
		S-M                     & $1.24$    & $0.90$    & $0.59$     & $0.15$    & $1.02$     & $0.91$     & $1.22$     & $0.35$    & $1.00$    & $0.31$    \\
		combo                   & $0.85$    & $1.00$    & $1.00$     & $0.92$    & $0.88$     & $0.96$     & $1.12$     & $0.98$    & $0.83$    & $1.00$    \\ \hline
	\end{tabular}
\caption{\textit{Comparison of the validation error curves.
		We average the estimation errors of the five validation methods in Table \ref{tab:method} at each tuning parameter under the five simulation settings in Table \ref{tab:scenario} and obtain validation error curves. We regress validation error curves over the oracle error curve. We present the coefficient, and $R^2$ of each regression.}}
\label{tab:simulation}
\end{table}

\subsection{Simulating from real data}
Real data analysis in causal inference is generally difficult, since the truth is unknown. 
Without an oracle, the aforementioned criteria: relative MSE and error curve similarity are infeasible.
To make one step towards real data analysis, we use features from a real dataset instead of generating covariates from an artificial distribution. 
Based on the real features, we generate treatment assignments and potential outcomes from model \eqref{eq:model}. 
In this way, the underlying truth is still tractable and comparisons between validation methods can be carried out. 
 
We consider the dataset of the SPRINT Data Analysis Challenge \cite{sprint2016systolic} launched by the New England Journal of Medicine. 
The dataset aims to study whether a new treatment program targeting reducing systolic blood pressure (SBP) will reduce cardiovascular disease (CVD) risk. 
There are $20$ features of interest: $3$ demographic features, such as age, race; $6$ medical history features, such as daily Aspirin use, history of CVD;  $11$ lab measurements, such as body mass index (BMI), SBP. 
We remove $3$ covariates due to spuriously high correlations. We match exactly on $6$ categorical covariates, and focus on the subgroup of white male with clinical or subclinical CVD history who are currently using statin and Aspirin. 
We ignore the covariate site since no significant batch effect is observed. 
Finally, we are left with $642$ valid observations and $10$ covariates. 

In each trial, we randomly sample two thirds of the units,  generate treatment assignments and responses under the combination of setting \RNum{3} and \RNum{5}, i.e. with confounding and model misspecification. 
The HTE estimator, validation methods and comparison criteria are the same as previous. 
Results are summarized in Figure \ref{fig:realData}.

We observe that the \textit{combo} method produces the most promising result, and the \textit{prd} method is not working favorably. 
The validation error curves of the methods other than the \textit{prd} method largely resemble the oracle error curve in trend, with the \textit{combo} method producing the most similar shape. 
In contrast, the error curve of the response prediction method does not capture the first-decrease-then-increase pattern, and is decreasing in the range of the tuning parameters considered.

\begin{figure}[]
    \centering
    \begin{minipage}{7cm}
    \centering  
\includegraphics[width = 6cm, height = 6cm]{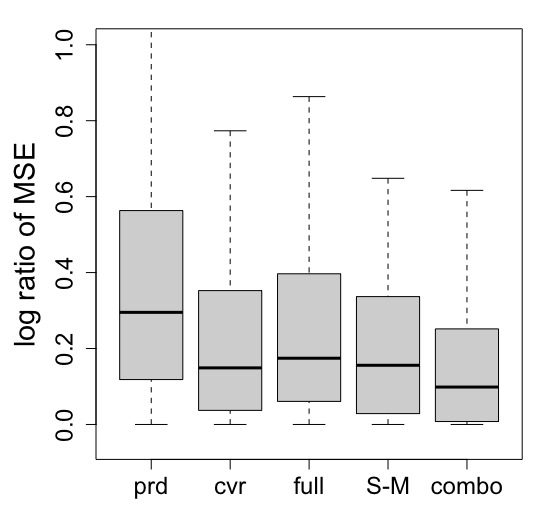}
    \subcaption*{\textit{(a) boxplot of log ratio of MSE}}
    \end{minipage}
    \begin{minipage}{7cm}
    \centering  
\includegraphics[width = 6.5cm, height = 6cm]{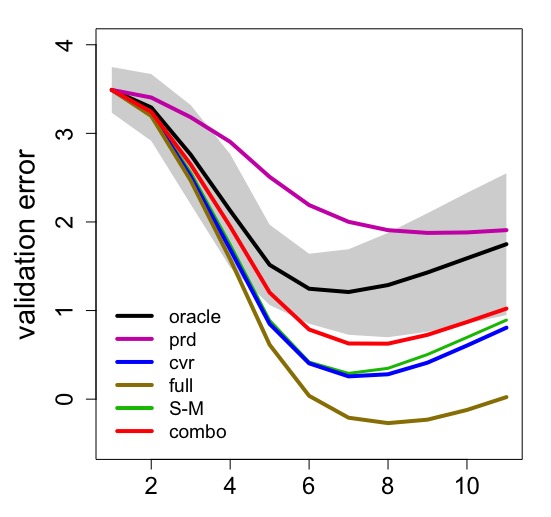}  
    \subcaption*{\textit{(b) validation error curves}}
    \end{minipage}
    \caption{\textit{Comparison of validation methods on the synthetic data generated from the dataset SPRINT. 
    In the left panel, we display the boxplot of the log ratio of $\text{MSE}_{\text{method}}$ over $\text{MSE}_{\text{oracle}}$ of the five validation methods in Table \ref{tab:method}. 
    In the right panel, we plot the averaged validation error curves of the five validation methods and the oracle error curve. 
    The x-axis plots the tuning parameter values $2^{-\frac{i}{2}}$ for $i \in \{1,2,\ldots, 11\}$, and the y-axis plots the validation errors.
    For visualization, we shift the validation error curves so that the starting points coincide.}}
    \label{fig:realData}
\end{figure}

\section{Extension to general exponential family}\label{sec:extension}
In the previous sections, we dealt with continuous responses. 
In real world, there are other types of outcomes worthwhile to study.
For instance, doctors study the effectiveness of a certain surgery by measuring whether the patients underwent the operation or not survive to a certain time spot; governments investigate the influence of a policy encouraging non-motor vehicles by comparing the times of bicycles used from automated bicycle counters before and after the policy is enforced.
In this section, we extend the aforementioned assessment approach to address multiple types of responses.

We generalize the model \eqref{eq:model} to general exponential family, which deals with a wide range of responses including binary data and count data. 
Mathematically, we assume
\begin{align}
\label{eq:generalizedModel}
\begin{split}
    Y_i | W_i, X_i &\stackrel{ind.}{\sim} \kappa(Y_i) \cdot \exp\left\{\eta(X_i, W_i) Y_i - \psi(\eta(X_i, W_i))\right\},
\end{split}
\end{align}
where $\eta(x, w)$ represents the natural parameter, $\psi(\eta)$ is the cumulant generating function, and $\kappa(y)$ is the carrying density.
More explicitly, we formulate the natural parameter as
\begin{align*}
     \eta(x, w) = 
     \begin{cases}
          \mu(x), \quad & w = 0,\\
          \nu(x), \quad & w = 1,
     \end{cases}
\end{align*}
and the treatment effect $\tau(x)$ is the difference in natural parameters of one unit under treatment and control.
The model \eqref{eq:generalizedModel} with Gaussian distribution is a sub-case of the original model \eqref{eq:model}.

Next, we generalize the validation criterion, i.e. the mean squared error in \eqref{eq:validationError}.
We first state the following result of conditional likelihood. 
\begin{proposition}
\label{prop:glm}
Consider $n$ pairs of data $\{(X_{i1}, X_{i2}, W_{i1}, W_{i2}, Y_{i1}, Y_{i2})\}$, where $\mu(X_{i1}) = \mu(X_{i2})$, $W_{i1} = 1$, $W_{i2} = 0$, and $Y_{i1}$, $Y_{i2}$ are generated independently from model \eqref{eq:generalizedModel} given $X_{ij}$, $W_{ij}$. 
Then the conditional likelihood of $\{Y_{ij}\}$ given $\{Y_{i1} + Y_{i2}\}$ does not depend on $\mu(x)$.
\end{proposition}
Proposition \ref{prop:glm} implies that with pairs agreeing on control group mean function values, the conditional likelihood --- which serves as a valid criterion for the HTE estimation assessment --- can be evaluated with no information of $\mu(x)$.
For example, consider the special case of Gaussian distribution, the log conditional likelihood equals $\sum_{i=1}^n (Y_{i2} - Y_{i1})^2$ up to scale, which agrees with \eqref{eq:validationError}.

If data comes in perfectly matched pairs, the condition $\mu(X_{i1}) = \mu(X_{i2})$ is automatically satisfied.
Examples of Proposition \ref{prop:glm} with perfectly matched pairs can be found in \cite{argesti1996introduction}.
When perfectly matched data are not available, we can apply the matching based on the proximity score distance to construct pairs such that $\mu(X_i) \approx \mu(X_j)$. 
Based on the matched pairs, we compute the conditional likelihood pretending the pairs are perfectly matched, and use the conditional likelihood as the criterion for model selection.

\section{Discussion}\label{sec:discussion}
In this paper, we propose an assessment approach of HTE estimation by constructing pseudo-observations based on matching. 
For matching, we propose to minimize the average proximity score distance.
When conducting the assessment approach under the cross-validation framework, we propose to match before split.

The assessment approach can be adapted for data calibration.
Given an estimator, a standard way of calibration tune the width of a prediction band on the hold-out data according to the coverage of observations.
As for the calibration of a HTE estimator, observation coverage is irrelevant.
Instead, we can construct pseudo-observations as discussed and determine the width of the prediction band by covering a certain proportion of the pseudo-observations.

A limitation of the assessment approach lies in the computation of matching. 
Solving a matching problem exactly is generally computationally heavy. 
Consider the simplest case where each treated unit is mapped to exactly one control unit and the distance matrix is not sparse, minimizing the total/average distance takes time $O(n^3)$.
Thus, fast approximate matching algorithms are desirable to make the validation method scalable.


\bibliographystyle{alpha}
\bibliography{ref.bib}

\section{Appendix}

\begin{proof}[Proof of Proposition \ref{prop:validationError}]
We prove for bias and variance respectively.

For bias, under model \eqref{eq:model}
\begin{align}
\label{eq:errorExpt}
\begin{split}
    \EE\left[\widehat{\text{error}_{\pi}}\right] 
    &= \EE\left[\frac{1}{|\Pi|}\sum_{(t_i, c_j) \in \Pi}(Y_{t_i} - Y_{c_j} - \hat{\tau}(X_{t_i}))^2\right]\\
    &= \frac{1}{|\Pi|}\sum_{(t_i, c_j) \in \Pi}\EE\left[(\mu(X_{t_i}) - \mu(X_{c_j}) + \tau(X_{t_i}) - \hat{\tau}(X_{t_i}) + \varepsilon_{t_i} - \varepsilon_{c_j})^2\right]\\
    &= \frac{1}{|\Pi|}\sum_{(t_i, c_j) \in \Pi}(b_{t_i,c_j} + \tau(X_{t_i}) - \hat{\tau}(X_{t_i}))^2 + 2\sigma^2\\
    &= \text{error}_{\pi}  + \frac{2}{|\Pi|}\sum_{(t_i, c_j) \in \Pi} b_{t_i,c_j} \cdot (\tau(X_{t_i}) - \hat{\tau}(X_{t_i})) + \overline{b^2_{\pi}} + 2\sigma^2.
\end{split}
\end{align}
By Cauchy-Schwarz inequality, 
\begin{align}
\label{eq:CauchySchwarz}
\begin{split}
    &~\frac{1}{|\Pi|}\sum_{(t_i, c_j) \in \Pi}b_{t_i, c_j}\cdot (\hat{\tau}(X_{t_i}) - \tau(X_{t_i})) \\
    &\le \left(\frac{1}{|\Pi|}\sum_{(t_i, c_j) \in \Pi}b_{t_i, c_j}^2 \right)^{\frac{1}{2}} \left(\frac{1}{|\Pi|}\sum_{(t_i, c_j) \in \Pi}(\hat{\tau}(X_{t_i}) - \tau(X_{t_i}))^2 \right)^{\frac{1}{2}} \\
    &= (\overline{b^2_{\pi}})^{\frac{1}{2}} \cdot (\text{error}_{\pi})^{\frac{1}{2}}.
\end{split}
\end{align}
Plug \eqref{eq:CauchySchwarz} into \eqref{eq:errorExpt}, and divide both sides by $\text{error}_{\pi}$,
\begin{align*}
    \frac{\EE\left[\widehat{\text{error}_{\pi}}\right]  - 2 \sigma^2}{\text{error}_{\pi}}
    \le 1 + 2 \sqrt{\frac{\overline{b^2_{\pi}}}{\text{error}_{\pi}}} + \frac{\overline{b^2_{\pi}}}{\text{error}_{\pi}}
    = \left(1 + \sqrt{\frac{\overline{b^2_{\pi}}}{\text{error}_{\pi}}}\right)^2.    
\end{align*}
Similarly for the lower bound of the bias.

For variance, let $\delta_{t_i, c_j} = \tau(X_{t_i}) - \hat{\tau}(X_{t_i})$. Define the neighborhood of a matched pair $(t_i, c_j)$ as
\begin{align*}
    \calN_{t_i, c_j}^{\pi} = \left\{(t_i', c_j') \in \Pi: t_i = t_i' \text{ or } c_j = c_j', ~(t_i', c_j') \neq (t_i, c_j) \right\}.
\end{align*}
Since each treated unit falls into at most $M_t$ pairs, and each control unit falls into at most $M_c$ pairs,
\begin{align}
\label{eq:neighborhoodSize}    
    \left|\calN_{t_i, c_j}^{\pi}\right|
    \le M_t + M_c - 2.
\end{align}
Under model \ref{eq:model},
\begin{align*}
    \var\left(\widehat{\text{error}_{\pi}}\right)
    &= \var\left(\frac{1}{|\Pi|}\sum_{(t_i, c_j) \in \Pi}(Y_{t_i} - Y_{c_j} - \hat{\tau}(X_{t_i}))^2\right)
    = \var\left(\frac{1}{|\Pi|}\sum_{(t_i, c_j) \in \Pi}(b_{t_i,c_j} + \delta_{t_i,c_j} + \varepsilon_{t_i} - \varepsilon_{c_j})^2\right)\\
    &= \frac{1}{|\Pi|^2} 
    \left(\sum_{(t_i, c_j) \in \Pi} \var\left((b_{t_i,c_j} + \delta_{t_i,c_j} + \varepsilon_{t_i} - \varepsilon_{c_j})^2\right)\right.\\
    & + \left.\sum_{(t_i, c_j) \in \Pi} \sum_{(t_i', c_j') \in \calN_{t_i, c_j}^{\pi}} \cov\left((b_{t_i,c_j} + \delta_{t_i,c_j} + \varepsilon_{t_i} - \varepsilon_{c_j})^2 , (b_{t_i',c_j'} + \delta_{t_i',c_j'} + \varepsilon_{t_i'} - \varepsilon_{c_j'})^2\right)\right).
    \end{align*}
    Since $2~\cov(\xi_1, \xi_2) \le \var(\xi_1) + \var(\xi_2)$  for random variables $\xi_1$, $\xi_2$,
    \begin{align*}
    \var\left(\widehat{\text{error}_{\pi}}\right)
    &\le \frac{1}{|\Pi|^2} 
    \left(\sum_{(t_i, c_j) \in \Pi} \var\left((b_{t_i,c_j} + \delta_{t_i,c_j} + \varepsilon_{t_i} - \varepsilon_{c_j})^2\right)\right.\\
    & + \left.\frac{1}{2} \sum_{(t_i, c_j) \in \Pi} \sum_{(t_i', c_j') \in \calN_{t_i, c_j}^{\pi}} \var\left((b_{t_i,c_j} + \delta_{t_i,c_j} + \varepsilon_{t_i} - \varepsilon_{c_j})^2 \right) + \var\left((b_{t_i',c_j'} + \delta_{t_i',c_j'} + \varepsilon_{t_i'} - \varepsilon_{c_j'})^2\right)\right)\\
    \end{align*}
    By \eqref{eq:neighborhoodSize},
    \begin{align*}
    \var\left(\widehat{\text{error}_{\pi}}\right)
    &\le \frac{1}{|\Pi|^2} 
    \left(\sum_{(t_i, c_j) \in \Pi} \var\left((b_{t_i,c_j} + \delta_{t_i,c_j} + \varepsilon_{t_i} - \varepsilon_{c_j})^2\right)
    + \sum_{(t_i, c_j) \in \Pi} \left|\calN_{t_i, c_j}^{\pi}\right| \cdot  \var\left((b_{t_i,c_j} + \delta_{t_i,c_j} + \varepsilon_{t_i} - \varepsilon_{c_j})^2 \right)\right)\\
    &\le \frac{1}{|\Pi|^2} 
    \left(\sum_{(t_i, c_j) \in \Pi} (M_t + M_c - 1) \cdot \var\left((b_{t_i,c_j} + \delta_{t_i,c_j} + \varepsilon_{t_i} - \varepsilon_{c_j})^2 \right)\right).
\end{align*}
Recall that $\var(\varepsilon) = \sigma^2$, $\var(\varepsilon^2) = \kappa \sigma^4$, 
\begin{align*}
    &~~\frac{1}{|\Pi|}\sum_{(t_i, c_j) \in \Pi} \var\left((b_{t_i,c_j} + \delta_{t_i,c_j} + \varepsilon_{t_i} - \varepsilon_{c_j})^2\right) \\
    &\le \frac{2}{|\Pi|} \sum_{(t_i, c_j) \in \Pi} \var\left(\left(\varepsilon_{t_i} - \varepsilon_{c_j}\right)^2\right)  + 4 \left(b_{t_i,c_j} + \delta_{t_i,c_j}\right)^2 \var\left(\varepsilon_{t_i} - \varepsilon_{c_j}\right)  \\
    &=\frac{2}{|\Pi|} \sum_{(t_i, c_j) \in \Pi} (2 \kappa + 4) \cdot \sigma^4 + 8\cdot \left(b_{t_i,c_j} + \delta_{t_i,c_j}\right)^2 \sigma^2 \\
    &\le (4 \kappa + 8) \cdot \sigma^4 + 32  \sigma^2 \cdot \left(\overline{b^2_{\pi}} + \text{error}_{\pi}\right).
\end{align*}
\end{proof}

\begin{proof}[Proof of Proposition \ref{prop:optimizationTarget}]
For $C > 0$, 
\begin{align*}
    \sum_{(t_i, c_j) \in \Pi}
    C \cdot d_{t_i, c_j} 
    = C \cdot \sum_{(t_i, c_j) \in \Pi}
    d_{t_i, c_j}, 
\end{align*}
thus the total distance optimization is invariant to scaling. The example in Figure \ref{fig:toyExample2} implies that the total distance minimization is not invariant to translation.

For $C_1$, $C_2 > 0$, 
\begin{align*}
    \frac{1}{|\Pi|}\sum_{(t_i, c_j) \in \Pi}
    C_1 \cdot d_{t_i, c_j} + C_2
    = C_2 + \frac{C_1}{|\Pi|} \sum_{(t_i, c_j) \in \Pi}
    d_{t_i, c_j}, 
\end{align*}
thus the average distance minimization is invariant to both scaling and translation.

Let $\pi_{\text{ave}}$, $\pi_{\text{tot}}$ be the optimal solution of average distance minimization. By the optimality condition,
\begin{align*}
    \left|\Pi_{\text{ave}}\right|
    = \frac{D_{\text{tot}}(\Pi_{\text{ave}})}{D_{\text{ave}}(\Pi_{\text{ave}})} 
    \ge \frac{D_{\text{tot}}(\Pi_{\text{tot}})}{D_{\text{ave}}(\Pi_{\text{tot}})}
    = \left|\Pi_{\text{tot}}\right|.
\end{align*}
\end{proof}

\begin{proof}[Proof of Proposition \ref{prop:glm}]
Since $Y_{i1}$, $Y_{i2}$ are generated independently from model \eqref{eq:generalizedModel}, the marginal density of $Z_i = Y_{i0} + Y_{i1}$ follows 
\begin{align*}
    f_{Z_i}(z) 
    &= \int f_{Y_{i0}}(y_0) \cdot f_{Y_{i1}}(z-y_0) d y_0  \\
    &= \exp\{-\psi(\eta_1) - \psi(\eta_0)\} \int \kappa(y_0) \kappa(z - y_0) \exp\{\eta_1 y_0 + \eta_0 (z-y_0)\} d y_0 \\
    &= \exp\{-\psi(\eta_1) - \psi(\eta_0)\} \cdot \exp\{\eta_0 z \} \int \kappa(y_0) \kappa(z - y_0) \exp\{(\eta_1 - \eta_0) y_0\} d y_0.
\end{align*}
Then the conditional distribution given $Z_i$ is
\begin{align}
\label{eq:general}
\begin{split}
&f\left(Y_{i0} = y_0, Y_{i1} = z-y_0 \mid Z_i = z\right) \\
=&~ \frac{\kappa(z-y_0) \exp\{\eta_1 (z-y_0) - \psi(\eta_1)\} \cdot \kappa(y_0) \exp\{\eta_0 y_0 - \psi(\eta_0)\}}{\exp\{-\psi(\eta_1) - \psi(\eta_0)\} \cdot \exp\{\eta_0 z \} \int \kappa(y_0) \kappa(z - y_0) \exp\{(\eta_1 - \eta_0) y_0\} d y_0} \\
=&~ \frac{\kappa(z-y_0) \kappa(y_0) \exp\{(\eta_1 - \eta_0) (z-y_0) \}}{\int \kappa(y_0) \kappa(z - y_0) \exp\{(\eta_1 - \eta_0) y_0\} d y_0}.
\end{split}
\end{align}
Since $\mu(X_{i0}) = \mu(X_{i1})$,
\begin{align*}
    \eta_1 - \eta_0 
    = \mu(X_{i1}) + \tau(X_{
    i1}) - \mu(X_{i0})
    = \tau(X_{i1}),
\end{align*}
thus the conditional likelihood does not depend on $\mu(x)$.
\end{proof}




\end{document}